\documentclass[twocolumn, journal]{IEEEtran}

\usepackage{booktabs} % For professional looking tables
\usepackage{multirow}
\usepackage{mathtools}

\usepackage{amssymb}

\newcommand{\hig}[1]{\textcolor{black}{#1}}

\usepackage{enumitem}

\usepackage{smartdiagram}
\usesmartdiagramlibrary{additions}

\graphicspath{{Figs/}}

\usepackage{pdfpages}
\usepackage{pdflscape}
\usepackage{amsmath}
\usepackage{amsthm}

\DeclareMathOperator*{\diag}{diag}

\usepackage{amsfonts}

\usepackage{color}

\usepackage{algorithm}
\usepackage{algorithmic}
\let\Algorithm\algorithm

\renewcommand\algorithm[1][]{\Algorithm[#1]\setstretch{1}}

\usepackage{setspace}
\usepackage{enumitem}
\usepackage{xcolor}

\usepackage{bm}
\usepackage{bbm}

\usepackage{tikz}
\usepackage{pgf}
\usetikzlibrary{arrows,automata}
\usetikzlibrary{shapes}
\tikzstyle{data44}=[rectangle split,rectangle split parts=2,draw,text centered]

\usepackage{subfigure}

\DeclareMathAlphabet{\mathcal}{OMS}{cmsy}{m}{n}

\usetikzlibrary{matrix}

\tikzset{
  BarreStyle/.style =   {opacity=.3,line width=14 mm,color=#1},
  node style ge/.style={},
  node style sp/.style={},
  yl/.style={},
  arrow style mul/.style={},
}

\newtheorem{theorem}{Theorem}
\newtheorem{lemma}{Lemma}

\newtheorem{remark}{Remark}
\newtheorem{definition}{Definition}

\usepackage{makecell}

\usepackage{threeparttable,booktabs}

\usepackage[colorlinks,
            linkcolor=blue,
            anchorcolor=blue,
            citecolor=blue
            ]{hyperref}

\usepackage{cite}

\usepackage{mdwlist}

\usepackage{colortbl}

\begin{document}

\title{Ensuring Network Connectedness in Optimal Transmission Switching Problems}
%Ensuring Network Connectedness in Optimal Transmission Switching
\author{Tong Han,
        Yue Song,~\IEEEmembership{Member},
        David J. Hill,~\IEEEmembership{Life Fellow}
%  \thanks{This work was supported by the Research Grants Council of the Hong Kong Special Administrative Region through the General Research Fund under Project No. 17209419.}
 \thanks{T. Han, Y. Song and D. J. Hill are with the Department of Electrical and Electronic Engineering, University of Hong Kong, Hong Kong (e-mail: hantong@eee.hku.hk, yuesong@eee.hku.hk, dhill@eee.hku.hk).}% <-this % stops a space
\vspace{-22pt}
}

% \markboth{Draft sent to Prof. David J.Hill}%
% {Shell \MakeLowercase{\textit{et al.}}: Bare Demo of IEEEtran.cls for IEEE Journals}

\maketitle

\begin{abstract}
  Network connectedness is indispensable for the normal operation of transmission networks. However, there still remains a lack of efficient constraints that can be directly added to the problem formulation of optimal transmission switching (OTS) to ensure network connectedness strictly. To fill this gap, this paper proposes a set of linear connectedness constraints by leveraging the equivalence between network connectedness and feasibility of the vertex potential equation of an electrical flow network. The proposed constraints are $\!$compatible with any existing OTS models to ensure topology connectedness. Furthermore, we develop a reduction version for the proposed connectedness constraints, seeking for improvement of computational efficiency. Finally, numerical studies with a DC OTS model show the deficiency of OTS formulations without full consideration of network connectedness and demonstrate the effectiveness of the proposed constraints. The computational burden caused by the connectedness constraints is moderate and can be remarkably relieved by using the reduced version.
\end{abstract}


\IEEEpeerreviewmaketitle

\section{Introduction}

\IEEEPARstart{T}{he} power system is the most important link in the country's energy system, which consists of generation, transmission and distribution systems. Optimal transmission switching (OTS) is the problem to find an optimal generation dispatch and transmission network topology to minimize the dispatch cost \cite{4-62}. Due to decreasing generation-side dispatchablility with growing penetration of variable renewable energy, OTS for leveraging grid-side flexibility is expected to be more widely and actively engaged in future network operations \cite{4-970}.

Network connectedness should be ensured for system normal operations \cite{4-0-1-ref-0}; however, this is not fully considered in most formulations of OTS. Ref. \cite{4-761} pointed out that previous OTS formulations can produce the optimal topology with unallowable islands and thus affecting network reliability. Before the work of \cite{4-761}, OTS formulations only contain necessary connectedness constraints or even no consideration of network connectedness \cite{4-627}. For this, \cite{4-761} designed a branching strategy to preserve network connectedness maximally during the branch-and-bound process. Nevertheless, it is still desirable to \textit{develop efficient constraints which can be directly added to OTS models to prevent islands strictly}. Until now, even though various forms of OTS, such as that considering $N\!-\!$1 security and using AC power flow, have been well developed, this gap is still not filled \cite{4-61, 4-735}. Accordingly, for small-scale networks or some particular cases, the optimal topology obtained by solving these OTS models can be connected, while for large-scale networks and more general cases, connectedness of the optimal topology is very likely to be unmet.

\hig{In this paper, we fill this long-existing gap with a twofold contribution. Firstly, a set of linear constraints is proposed to strictly ensure network connectedness in OTS problems. These constraints are in terms of the vertex potential condition over an auxiliary electrical flow network. The proposed constraints can be directly added to any existing OTS models without changing their properties. Secondly, by exploiting the equivalence between these constraints and partial OTS constraints for certain subgraphs and the fact that not all lines are switchable in OTS, a reduced version of the proposed connectedness constraints seeking for improvement of computational efficiency is also developed.} The effectiveness and computational efficiency of the proposed approach are also demonstrated numerically.

\section{Optimal Transmission Switching}

\subsection{Notation and Model of Optimal Transmission Switching}

We first introduce notations used hereafter and the model of OTS to make the paper self-contained. The transmission network containing $N_n$ nodes and $N_e$ branches (including lines and transformers) is represented as an undirected connected multigraph $\mathcal{G}(\mathcal{V}, \mathcal{E}, r)$. The node set $\mathcal{V} \!\!=\!\! \{i |_{i=1}^{N_n}\}$ and edge set $\mathcal{E} \!\!=\!\! \{k |_{k=1}^{N_e} \}$ correspond to all nodes and all branches, respectively 
\footnote{For brevity, denote $\{\!1,\!2,\!..., n\}$ by $\{\!i|_{i=1}^n\!\}$ and $\{\!x_1,\!x_2,\!..., x_n\!\}$ by $\{\!x_i\!|_{i=1}^n\!\}$. }; and $r:\! \mathcal{E} \!\!\to\!\! \{(i,j) | {i,j} \!\in\! \mathcal{V} \}$ is the function from $\mathcal{E}$ to the set of all pairs of two elements of $\mathcal{V}$ so that $r(k)$ with $k \!\in\! \mathcal{E}$ determines the two nodes linked by branch $k$. In graph notations hereinafter, map $r$ is omitted for brevity. Let $\bm{z} \!\!=\!\! [z_k]_{k \in \mathcal{E}} \!\in\! \mathbb{B}^{N_e} $ be the vector of binary variables to represent status of branches, where ${z}_k \!\!=\!\! 1$ if branch $k$ is switched on and $z_k \!\!=\! 0$ otherwise. We use $\mathcal{G}_{\bm{z}}$ to denote the edge-induced subgraph of $\mathcal{G}$ by edges $\{\!k \!\in\! \mathcal{E} | z_k \!\!=\!\! 1 \}$, which corresponds to the transmission network after line switching assigned by $\bm{z}$. 
%The edge-induced subgraph is a subset of the edges of graph $\mathcal{G}$ together with all of the nodes that are their endpoints. 
With each edge of $\mathcal{G}$ assigned an arbitrary and fixed orientation, denote by $\bm{E}_{\mathcal{G}}$ and $\bm{E}_{\mathcal{G}_{\bm{z}}}$ the oriented incidence matrices of $\mathcal{G}$ and $\mathcal{G}_{\bm{z}}$, respectively. Denote by $\bm{L}_{\mathcal{G}_{\bm{z}}}$ the Laplacian matrix of graph $\mathcal{G}_{\bm{z}}$.

A typical model of OTS with DC power flow is formulated as a mixed-integer linear program (MILP) as follows \cite{4-62, 4-64}:
\vspace{-13pt}
\begin{subequations}\label{eq-0-1-0}
  \begin{align}
    &  \!\!\!\!\!\!\!\!\!\! \min\nolimits_{\bm{z} \in \mathbb{B}^{N_e} ,\bm{p}_g \in \mathbb{R}^{N_n}  }  f(\bm{p}_g, \bm{z})   \label{eq-0-1-0:1}\\ 
     \text{s.t.} ~& \theta_i^{\min} \leq \theta_i \leq {\theta}_i^{\max}, {p}_{g,i}^{\min} \leq p_{g,i} \leq {p}_{g,i}^{\max}   ~~ \forall i \in \mathcal{V} \label{eq-0-1-0:2} \\
     & - {p}_{b,k}^{\max} z_{k} \leq p_{b, k} \leq {p}_{b, k}^{\max} z_{k} ~~ \forall k \in \mathcal{E} \label{eq-0-1-0:5} \\
     & b_k (\theta_i \!-\! \theta_j) \!-\! p_{b,k} \!\!+\!\! (1 \!-\! z_k) K \!\!\geq\!\! 0 ~~ \forall k \!\in\! \mathcal{E}, (i,j) \!\!=\!\! r(k)  \label{eq-0-1-0:6} \\
     & b_k (\theta_i \!-\! \theta_j) \!-\! p_{b,k} \!\!-\!\! (1 \!-\! z_k) K \!\!\leq\!\! 0 ~~ \forall k \!\in\! \mathcal{E}, (i,j) \!\!=\!\! r(k) \label{eq-0-1-0:7} \\
     & p_{g, i} \!-\! p_{d, i} \!-\!\!\! \textstyle \sum\nolimits_{k \in \mathcal{A}_{\!f}\!(i)  } \!p_{b, k} \!\!+\!\! \textstyle  \sum\nolimits_{k \in \mathcal{A}_t(i)  } \!p_{b, k} \!\!=\!\! 0 ~~ \forall i \!\in\! \mathcal{V} \label{eq-0-1-0:4} \\
     & z_k = 1 ~~ \forall k \in \mathcal{E}/\mathcal{E}_{s} \label{eq-0-1-0:8}  
  \end{align}
\end{subequations}
where $\bm{p}_g \!\!=\!\! [p_{g,i}]_{i \in \mathcal{V}}$ with $p_{g,i}$ being the real power generation at node $i$; $\theta_i$ is the phase angle of voltage at node $i$, with $\theta_i^{\min}$ and $\theta_i^{\max}$ being it lower and upper bounds, respectively; $p_{g,i}^{\min}$ and $p_{g,i}^{\max}$ are lower and upper bounds of real power generation at node $i$, respectively; $p_{b,k}$ is the real power flow from node $i$ to $j$ through branch $k$, where $(i,j) \!\!=\!\! r(\!k\!)$; $p_{d,i}$ denotes the real power consumed at node $i$, which are given by load forecast; $\mathcal{A}_f(i) \!\!\subset\!\! \mathcal{E} $ and $\mathcal{A}_t(i) \!\!\subset\!\! \mathcal{E} $ are sets of all branches starting and ending at node $i$, respectively; $p_{b,k}^{\max}$ is the maximal capacity of branch $k$; $b_k$ denotes the susceptance of branch $k$; $K$ is a sufficiently large positive number; and $\mathcal{E}_s$ is the set of lines allowed to be switched, and then $\mathcal{E}_u \!\!=\!\! \mathcal{E}/\mathcal{E}_{\!s}$ is the set of unswitchable branches including transformers and lines that do not participate in OTS. Objective function $f(\bm{p}_g, \bm{z})$ denotes the dispatch cost, which is generally formulated as the total generation cost, i.e., $\sum\nolimits_{i \in \mathcal{V}} \! c_{g,i} p_{g,i} $ with $c_{g,i}$ being the cost of power generation at node $i$, and the dispatch cost for branches, i.e, $\sum_{k\in \mathcal{E}_s} \!\! c_{b, k} z_k$ with $c_{b, k}$ being the transmission cost of branch $k$, can also be included \cite{4-63}; (\ref{eq-0-1-0:2}), (\ref{eq-0-1-0:5}) and (\ref{eq-0-1-0:8}) are operational constraints; and (\ref{eq-0-1-0:6})-(\ref{eq-0-1-0:4}) are DC power flow constraints. \hig{More practically, $N\!-\!1$ security constraints can be added to (\ref{eq-0-1-0}), which gives the $N\!-\!1$ security-constrained DC OTS model \cite{4-627, 4-63}. Since this work focuses on the connectedness issue in OTS, the following $N\!-\!1$ security constraints considering only branch contingencies are adopted:}
\begin{equation}\label{eq-0-0-add-1}
    \hig{ 
        (\bm{z} \circ \bm{s}^{\kappa}, \bm{p}_g^{\kappa}) \in \mathbb{F}^{\kappa}, \bm{r}_-  \leq \bm{p}_g^{\kappa} - \bm{p}_g \leq \bm{r}_+  ~~ \forall \kappa \in \mathcal{C}
    }
\end{equation}
\hig{where $\mathcal{C} \!\subseteq\! \mathcal{E}$ is the branch contingency set, $\kappa$ denotes the fault branch, $\bm{s}^{\kappa} \!=\! [s^{\kappa}_k]_{k \in \mathcal{E}}$ with $s^{\kappa}_k \!=\! 0$ if $k=\kappa$ and $s^{\kappa}_k \!=\! 1$ otherwise is the vector to parameterize contingency $\kappa$, $\bm{z} \circ \bm{s}^{\kappa}$ represents the post-contingency status of branches, $\bm{p}_g^{\kappa}$ is the vector of real power generation after contingency $\kappa$, the first constraint in (\ref{eq-0-0-add-1}) represents the post-contingency operational feasibility with $\mathbb{F}^{\kappa}$ being analogous to the feasible space of $(\bm{z}, \bm{p}_g)$ defined by (\ref{eq-0-1-0:2}) to (\ref{eq-0-1-0:4}), the second one in (\ref{eq-0-0-add-1}) represents ramping constraints of generators with $\bm{r}_+$ and $\bm{r}_-$ being vectors of upward and downward ramp rate, respectively.} It is noted that DC power flow here can be replaced by accurate AC power flow. We use DC OTS for illustration since accuracy of OTS models is not the focus of this paper while the proposed approach is also applicable to AC OTS.

\subsection{Connectedness Guarantee by (1) and Necessary Conditions}

Connectedness of graph $\mathcal{G}_{\bm{z}}\!$ is indispensable for the system normal operation, which, however, cannot be always ensured by model (\ref{eq-0-1-0}). An obvious example where model (\ref{eq-0-1-0}) potentially fails to ensure network connectedness of $\mathcal{G}_{\bm{z}}\!$ is illustrated in Fig. \ref{fig-0-1-3}. 
Solution \textit{A} represents a feasible solution of (\ref{eq-0-1-0}) with the minimum total generation cost, where the two green lines form a cutset of graph $\mathcal{G}_{\bm{z}}\!$ that corresponds to a partition of the power network into Area 1 and Area 2. Total real power generation of Area 1 and Area 2 are $p_{g, a1}$ and $p_{g, a2}$, respectively; and real power of $p_{b,c}$ is transferred from Area 1 to Area 2 through the cutset. Now we consider a power generation change where the real power generation of Area 1 decreases to $p_{g, a1} \!-\! p_{b,c}$ and that of Area 2 increases to $p_{g, a2} \!+\! p_{b,c}$. Then provided that operational constraints (\ref{eq-0-1-0:2}) and (\ref{eq-0-1-0:5}) are still satisfied after the power generation change and the decrease in generation cost of Area 1 equals to the increase in generation cost in Area 2, solution \textit{B} that contains two islands is also a feasible solution of (\ref{eq-0-1-0}) with the minimum total generation cost. 
Further consider two cases of the objective function in (\ref{eq-0-1-0}). 
When $f \!\!=\!\! \sum\nolimits_{i \in \mathcal{V}} c_{g,i} p_{g,i}$, the two solutions in Fig. \ref{fig-0-1-3} are both optimal solutions of (\ref{eq-0-1-0}). Solving (\ref{eq-0-1-0}) yields either the connected solution \textit{A} or unconnected solution \textit{B}. % \cite{4-855-1}.
When $f \!\!=\!\! \sum\nolimits_{i \in \mathcal{V}} c_{g,i} p_{g,i} \!+\!\! \sum_{k\in \mathcal{E}_s} \!c_{b, k} z_k$, solving (\ref{eq-0-1-0}) will certainly return the unconnected solution as long as $c_{b, k} > 0$ for the two lines in the cutset. Furthermore, even if the generation cost of solution \textit{B} is higher than that of solution \textit{A}, solving (\ref{eq-0-1-0}) will also give the unconnected solution if the difference between generation cost of the two solutions is smaller than the sum of $c_{b, k}$ of the two lines in the cutset.

\begin{figure}[h]
	\centering
	\includegraphics[width=7.5cm]{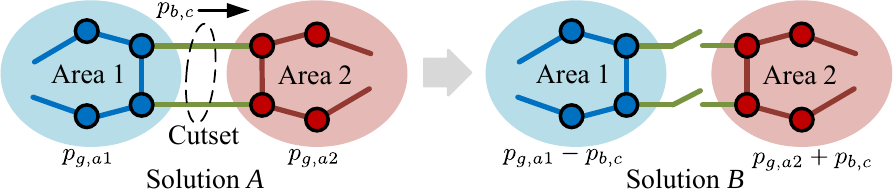}
	\caption{\hig{Illustration of an islanded optimal topology obtained from (\ref{eq-0-1-0}).}}
	\label{fig-0-1-3}
\end{figure}

In fact, when $f$ is the total generation cost and optimization model (\ref{eq-0-1-0}) has a unique optimal solution of $\bm{z}$, there is no guarantee that the unique optimal topology has to be connected. A necessary connectedness condition used in \cite{4-627} is 
\begin{equation}\label{eq-0-1-nc}
  \textstyle  \sum\nolimits_{ k \in \mathcal{A}_f(i) \cup \mathcal{A}_t(i) } z_k \geq 1 ~~ \forall i \in \mathcal{V} 
\end{equation}
which ensures that each node is linked by at least one branch. Considering the example in Fig. \ref{fig-0-1-3} again, constraints (\ref{eq-0-1-nc}) can only prevent islands when Area 1 or Area 2 contains one node.

Although the existing OTS models contain no particular constraints to guarantee network connectedness or only use necessary conditions, the obtained optimal topology can still be connected in most cases. This phenomenon can be explained considering the following points:
\begin{itemize}
  \item[\textit{(\romannumeral1)}] The number of lines being switching off is bounded by a relatively small number.
  \item[\textit{(\romannumeral2)}] The test power systems are highly meshed.
  \item[\textit{(\romannumeral3)}] The set of lines allowed to be switched is restricted.
  \item[\textit{(\romannumeral4)}] The $N\!-\!1$ security constraints are added to the OTS model.
  \item[\textit{(\romannumeral5)}] The numerical tests are based on system data at limited sample points.
\end{itemize}
The above points all can reduce the probability that $\mathcal{G}_{\bm{z}}$ with islands is feasible for the OTS model. We also use the example in Fig. \ref{fig-0-1-3} to explain this. It is trivial that points (\textit{\romannumeral1}) to (\textit{\romannumeral3}) can prevent producing solution \textit{B}. 
% For point (\textit{\romannumeral1}), assuming that constraint $\sum\nolimits_{k \in \mathcal{E}} ( 1 \!-\! z_k ) \!\leq\! 3$ is added to (\ref{eq-0-1-0}) to bound the number of lines being switched off and two lines in Area 1 are already switched off, solution \textit{B} will become infeasible. For point (\textit{\romannumeral2}), if two lines are further added to link Area 1 and Area 2, disconnecting the two area will violate constraint $\sum\nolimits_{k \in \mathcal{E}} ( 1 \!-\! z_k ) \!\leq\! 3$ even if no other lines are switched off. For point (\textit{\romannumeral3}), restricting any one line in the cutset to being unswitchable will prevent producing solution \textit{B}.
For point \textit{(\romannumeral4)}, if values of some electrical variables in solution \textit{B} are close to their limits, then a random line failure will very likely cause violations of operational constraints, namely that $N\!-\!1$ security constraints are not satisfied. Hence solving (\ref{eq-0-1-0}) will not produce solution \textit{B} in this case. For point \textit{(\romannumeral5)}, since the optimal topology from (\ref{eq-0-1-0}) depends on specific operating conditions, loads and generator costs, numerical tests in most existing works, which utilize IEEE standard system data with only one sample point or consider several load conditions, can miss the unconnected optimal topologies.

However, points \textit{(\romannumeral1)} and \textit{(\romannumeral3)} may increase system dispatch cost of the optimal solution. According to the computation results in \cite{4-62}, the best solution found for the IEEE 118-bus system opens 38 lines. Restrictions in \textit{(\romannumeral1)} and \textit{(\romannumeral3)} can prevent opening all these lines easily. Regarding point \textit{(\romannumeral2)}, although modern transmission networks are general highly meshed overall, they also contain some loosely-connected portions \cite{4-1125}. In fact, some transmission networks such as the south east Australian transmission network, are even loosely-connected on the whole \cite{4-118}. \hig{For the $N\!-\!1$ security constraints in \textit{(\romannumeral4)}, they may prevent islands in the optimal topology but connectedness of post-contingency graph $\mathcal{G}_{\bm{z} \circ \bm{s}^{\kappa}  }$ should be ensured additionally.}
As for point \textit{(\romannumeral5)}, for transmission networks with high penetration of renewable energy sources and in a market environment, not only power loads but also some power generation bounds and cost of power generation in the OTS model may vary greatly. Thus solving (\ref{eq-0-1-0}) without strict network connectedness constraints will produce unconnected optimal topologies in the long term almost certainly.

\section{Approach to Ensure Network Connectedness}

\subsection{Connectedness Constraints}

The following developes a set of constraints which can be directly added to OTS models to ensure network connectedness strictly. The general idea is to construct an auxiliary electrical flow network with the same topology as $\mathcal{G}_{\bm{z}}$ and existence of its electrical flow is guaranteed iff $\mathcal{G}_{\bm{z}}$ is connected. In this way, determination of connectedness of $\mathcal{G}_{\bm{z}}$ can be converted to that of feasibility of the vertex potential equation of the auxiliary electrical flow network, \hig{and the latter condition can be embedded into OTS models much more easily.}

We first introduce an electrical flow network (allow multiple edges between two vertices) with the same topology as $\mathcal{G}_{\bm{z}}$, unit resistance for each edge, vertex potential denoted as $\bm{\vartheta} \!\in\! \mathbb{R}^{\! N_n }$, and vertex electrical flow injection given by $\bm{c} \!\in\! \mathbb{R}^{\!N_n}$ \cite{4-p-0-1-1}. 
The Laplacian matrix links $\bm{\vartheta}$ and $\bm{c}$ as $\bm{L}_{\mathcal{G}_{\bm{z}}} \bm{\vartheta} = \bm{c} $. 
Suppose that $\mathcal{G}$ has $N_{s}$ connected node-induced subgraphs (NISs) and their node sets are collected by $\{ \mathcal{V}_i |_{i = 1}^{N_{s}} \}$. Here a node-induce subgraph refers to an arbitrary nonempty subset of the nodes of graph $\mathcal{G}$ together with all of the edges whose endpoints are both in this subset. In such manner, $\{ \mathcal{V}_i |_{i = 1}^{N_{s}} \}$ includes node sets of connected components of all possible $\mathcal{G}_{\bm{z}}$. Without loss of generality, it is assumed that $\mathcal{V}_1 \!\!=\!\! \mathcal{V}$ is the node set of the NIS equal to $\mathcal{G}$. Let $\bm{J} \!\in\! \mathbb{R}^{ N_s \times N_n } $ be the constant matrix satisfying $\forall  i \!\in\! \{\! i' |_{i' = 1}^{N_{s}} \!\}$, $\bm{J}_{ij} \!\!=\!\! 1$ if $j \!\!\in\!\! \mathcal{V}_j$ and $\bm{J}_{ij} \!=\! 0$ otherwise. Fig. \ref{fig-0-1-2} illustrates set $\{\! \mathcal{V}_i |_{i = 1}^{N_{s}} \!\}$ and the corresponding matrix $\bm{J}$ for a 4-node graph. With the matrix $\bm{J}\!$, some balance properties of $\bm{c}$ are defined by Definition \ref{cond-0-1-1}.
 
\begin{figure}[h]
	\centering
	\includegraphics[width=7cm]{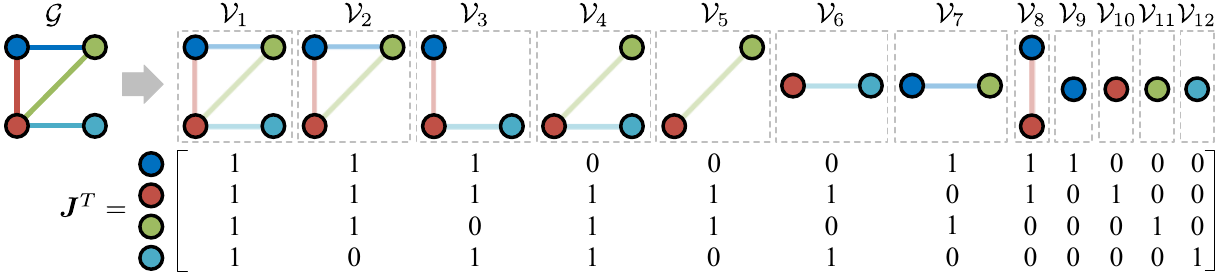}
	\caption{Illustration of set $\{ \mathcal{V}_i |_{i = 1}^{N_{s}} \}$ and matrix $\bm{J}$.}
	\label{fig-0-1-2}
\end{figure} 

\begin{definition}[Unique balance, unbalance]\label{cond-0-1-1}
  Let $\bm{b} \!\!=\!\! \bm{J} \bm{c} \!\in\! \mathbb{R}^{N_s} $. Then $\bm{c}$ is multiply-balanced if $\forall i \!\in\! \{i' |_{i'=1}^{N_m} \}$, $\bm{b}_i \!=\! 0$ and $\forall i \!\in\! \{i' |_{i'=N_m + 1}^{N_s} \}$, $\bm{b}_i \!\!\neq\!\! 0$, \hig{with $1 \!\!\leq\!\! N_m \!\!\leq\!\! N_s$.} If $N_m \!\!=\!\! 1$, $\bm{c}$ is uniquely-balanced. Moreover, $\bm{c}$ is unbalanced if $\forall i \!\in\! \{i' |_{i'=1}^{N_s} \}$, $\bm{b}_i \!\neq\! 0$.
\end{definition}

\begin{lemma}\label{lem-0-1-1}
  Given any uniquely-balanced $\bm{c}$, solutions of the vertex potential equation of the electrical flow network, i.e., $\bm{L}_{\mathcal{G}_{\bm{z}}} \bm{\vartheta} \!=\! \bm{c} $, exist iff graph $\mathcal{G}_{\bm{z}}$ is connected.
\end{lemma}

\begin{proof}
  Sufficiency. Suppose graph $\mathcal{G}_{\bm{z}}$ is connected. A necessary and sufficient condition for any solution(s) of $\bm{L}_{\mathcal{G}_{\bm{z}}} \bm{\vartheta} \!=\! \bm{c} $ to exist is that $\bm{W} \bm{c} = \bm{c}$ where $\bm{W} = \bm{L}_{\mathcal{G}_{\bm{z}}} \bm{L}_{\mathcal{G}_{\bm{z}}}^+$ with $\bm{L}_{\mathcal{G}_{\bm{z}}}^+$ being the Moore-Penrose pseudoinverse of $\bm{L}_{\mathcal{G}_{\bm{z}}}$. 

  Since $\mathcal{G}_{\bm{z}}$ is connected, by \cite[Lemma~3]{4-969}, $\bm{W}$ is given as $\bm{W}_{ij} = - \frac{1}{ N_n }$ for $i \neq j$ and $\bm{W}_{ij} =  \sum_{ j'=1, j' \neq j }^{N_n} \bm{W}_{ij'}  = \frac{N_n - 1}{ N_n }$ for $i = j$. Recall that $\bm{c}$ is uniquely-balanced, thus we have $\bm{1}_{N_n}^T \bm{c} = 0$. Therefore, $\forall i \in \{ i'|_{i'=1}^{N_n} \} $, $ [\bm{W} \bm{c}]_i \!=\! \sum_{j=1}^{N_n}  \bm{W}_{ij} \bm{c}_j \!=\! \bm{W}_{ii} \bm{c}_i \!+\! \sum_{j=1, j\neq i}^{N_n} \bm{W}_{ij} \bm{c}_j \!=\!  \frac{N_n - 1}{ N_n } \bm{c}_i \!+\! \frac{1}{ N_n } \bm{c}_i \!=\! \bm{c}_i  $, i.e., $\bm{W} \bm{c} \!=\! \bm{c}$. Thus solutions of $\bm{L}_{\mathcal{G}_{\bm{z}}} \bm{\vartheta} \!=\! \bm{c} $ exist.

  Necessity. Suppose that graph $\mathcal{G}_{\bm{z}}$ is disconnected and with $N_d$ connected components of size $n_k$, $\forall k \in \{ k'|_{k'=1}^{N_d} \}$. By rearranging nodes of each connected component together, we have $\bm{L}_{\mathcal{G}_{\bm{z}}} = \diag\{ \bm{L}_{\mathcal{G}_{\bm{z}}}^k |_{k=1}^{N_d} \} $ with $\bm{L}_{\mathcal{G}_{\bm{z}}}^k$ being the Laplacian matrix of the $k$-th connected component of $\mathcal{G}_{\bm{z}}$. Let $\bm{\vartheta}^k$ and $\bm{c}^k$ be subvectors of $\bm{\vartheta}$ and $\bm{c}$ corresponding to the $k$-th component, respectively; and $\bm{W}^k \!=\! \bm{L}_{\mathcal{G}_{\bm{z}}}^k\bm{L}_{\mathcal{G}_{\bm{z}}}^{k +} $. Consider existence of solutions of $\bm{L}_{\mathcal{G}_{\bm{z}}}^k \bm{\vartheta}^k \!=\! \bm{c}^k$. Following the proof of sufficiency, but since $\bm{1}_{n_k}^T \bm{c}^k \neq 0 $ by its unique balance, it is trivial that $\bm{W}^k \bm{c}^k \neq\bm{c}^k$. Thus $\forall k \!\in\! \{ k'|_{k'=1}^{N_d} \} $, solutions of $\bm{L}_{\mathcal{G}_{\bm{z}}}^k \bm{\vartheta}^k \!=\! \bm{c}^k$ do not exist and the same for $\bm{L}_{\mathcal{G}_{\bm{z}}} \bm{\vartheta} \!=\! \bm{c} $. 
\end{proof}

\begin{theorem}\label{theo-0-1-1}
  Graph $\mathcal{G}_{\bm{z}}$ is connected iff the following set of constraints is feasible: 
  \begin{subequations}\label{eq-0-1-1} 
    \begin{align}
      & - M (\bm{1}_{N_e} - \bm{z} ) \leq \bm{E}_{\mathcal{G}}^T \bm{\vartheta} - \bm{\rho}  \leq M (\bm{1}_{N_e} - \bm{z} )  \label{eq-0-1-1:0} \\  
      & - M \bm{z} \leq \bm{\rho} \leq M \bm{z}   \label{eq-0-1-1:2} \\ 
      & \bm{E}_{\mathcal{G}} \bm{\rho} = \bm{c}  \label{eq-0-1-1:3}
    \end{align}
  \end{subequations}
  where $\bm{\rho} \in \mathbb{R}^{ N_e }$ are auxiliary variables, $M$ is a sufficiently large positive number, and $\bm{c}$ is any uniquely-balanced one.
  % is any constant number satisfying $M \geq \Vert  \bm{Jc} \Vert_{\infty}$. 
\end{theorem}

\begin{proof}
  We first prove that $\bm{L}_{\mathcal{G}_{\bm{z}}} \bm{\vartheta} \!=\! \bm{c} $ and (\ref{eq-0-1-1}) are equivalent.

  First of all, we have $\bm{L}_{\mathcal{G}_{\bm{z}}} \bm{\vartheta} \!=\! \bm{c} $ $\Leftrightarrow$ $\bm{E}_{\mathcal{G}_{\bm{z}}} \bm{E}_{\mathcal{G}_{\bm{z}}}^T \bm{\vartheta} = \bm{c} $ $\Leftrightarrow$  $\{ \bm{E}_{\mathcal{G}_{\bm{z}}}^T \bm{\vartheta} = \bm{\rho}, \bm{E}_{\mathcal{G}_{\bm{z}}} \bm{\rho} = \bm{c} \}$ $\Leftrightarrow$  
  \begin{subequations}\label{eq-0-1-2} 
    \begin{align}
      & \bm{z} \circ ( \bm{E}_{\mathcal{G}}^T \bm{\vartheta}) = \bm{\rho} \label{eq-0-1-2:0} \\ 
      & \bm{E}_{\mathcal{G}} (\bm{z} \circ \bm{\rho} ) = \bm{c}. \label{eq-0-1-2:1}
    \end{align}
  \end{subequations}

  Regarding the left-hand side of (\ref{eq-0-1-2:0}) as bilinear terms of $\bm{z}$ and $\bm{E}_{\mathcal{G}}^T \bm{\vartheta}$, (\ref{eq-0-1-1:0}) and (\ref{eq-0-1-1:2}) are just the McCormick envelopes of (\ref{eq-0-1-2:0}), which are exact since $\bm{z} \in \mathbb{B}^{N_e}$. For bilinear terms $\bm{z} \circ \bm{\rho} $ in (\ref{eq-0-1-2:1}), we have $\bm{z} \circ \bm{\rho} = \bm{\rho}$ deriving from their exact McCormick envelopes with upper and lower bounds of $\bm{\rho}$ given by (\ref{eq-0-1-1:2}). Thus, (\ref{eq-0-1-2:1}) and (\ref{eq-0-1-1:3}) are equivalent if (\ref{eq-0-1-1:2}) holds. Then we conclude equivalence between $\bm{L}_{\mathcal{G}_{\bm{z}}} \bm{\vartheta} \!=\! \bm{c} $ and (\ref{eq-0-1-1}), which together with Lemma \ref{lem-0-1-1} gives Theorem \ref{theo-0-1-1}.
\end{proof}

% \begin{remark}
%   (\ref{eq-0-1-1}) are similar to power flow constraints together with branch flow limit constraints used in most optimization models for DC OTS, but where $\bm{c}$ is to be optimized and $M$ in (\ref{eq-0-1-1:2}) is replaced by terms related to branch flow limits.
% \end{remark}
 
According to Theorem \ref{theo-0-1-1}, by adding (\ref{eq-0-1-1}) with auxiliary variables $\bm{\vartheta} \!\in\! \mathbb{R}^{N_n}$ and $\bm{\rho} \!\in\! \mathbb{R}^{N_{\!e}} $ to DC OTS model (\ref{eq-0-1-0}), network connectedness will be guaranteed. In fact, constraints (\ref{eq-0-1-1}) is applicable to any optimization models of OTS including AC OTS \cite{4-61}. \hig{For the $N-1$ security-constrained OTS model with (\ref{eq-0-0-add-1}), (\ref{eq-0-1-1}) with $\bm{z}$ substituted by $\bm{z} \circ \bm{s}^{\kappa}$ should be introduced for each $\kappa \!\in\! \mathcal{C}$.} 
Notably, constraints (\ref{eq-0-1-1}) are all linear and without new binary variables being introduced. Thus constraints (\ref{eq-0-1-1}) will not change the type of original optimization problems. For the setting of $\bm{c}$ required to be uniquely-balanced, a simple way is letting any one $\bm{c}_i$ be $1 \!-\! N_n$ and the others be $1$.

 \begin{remark}\label{remark-0-1-1}
  Constraints (\ref{eq-0-1-1}) are similar to the set of constraints (\ref{eq-0-1-0:5})-(\ref{eq-0-1-0:4}) in the DC OTS model, but where $\bm{c}$ is to be optimized and $M$ in (\ref{eq-0-1-1:2}) is replaced by terms related to branch capacity.
\end{remark}

 \begin{remark}
   The proposed connectedness constraints can be easily extended to more general cases (e.g., interconnected microgrids) where power grids are also allowed to operate as certain multiple isolated sub-networks. In this case, we only need to replace the uniquely-balanced $\bm{c}$ by a multiply-balanced $\bm{c}$ with $\{ \mathcal{V}_i |_{i=2}^{N_m} \}$ corresponding to all node sets of allowable sub-networks.
 \end{remark}

\subsection{Reduction of the Connectedness Constraints}

We further reduce constraints (\ref{eq-0-1-1}) to improve computational efficiency. Firstly, according to Remark \ref{remark-0-1-1}, connectedness of some subgraphs of $\mathcal{G}_{\bm{z}}$ can be already ensured by (\ref{eq-0-1-0}) and thus corresponding constraints in (\ref{eq-0-1-1}) become redundant. To exploit this for the reduction, denote by $\mathcal{G}_{s}(\mathcal{V}_{s}, \mathcal{E}_{s} \!) $ an arbitrary connected NIS of $\mathcal{G}(\mathcal{V}, \mathcal{E})$. Let $\bm{p}_{ds} \!\!=\!\! [p_{d,i}]_{i \in \mathcal{V}_s}$, $\bm{p}_{gs} \!\!=\!\! [p_{g,i}]_{i \in \mathcal{V}_s}$, $\bm{p}_{gs}^{\min} \!\!=\!\! [p_{g,i}^{\min}]_{i \in \mathcal{V}_{\!s}}$, and $\bm{p}_{gs}^{\max} \!\!=\!\! [p_{g,i}^{\max}]_{i \in \mathcal{V}_s}$. Then the unbalanced connected NIS is defined by Definition \ref{def-0-1-partition}.
\begin{definition}[Unbalanced connected NIS]\label{def-0-1-partition}
  If $\bm{\varrho} \!=\! \bm{p}_{gs} \!-\! \bm{p}_{ds} $ is unbalanced for any $\bm{p}_{gs}$ satisfying $\bm{p}_{gs}^{\min} \!\leq\! \bm{p}_{gs} \!\leq\! \bm{p}_{gs}^{\max}$, then $\mathcal{G}_{s}( \mathcal{V}_{s}, \mathcal{E}_{s} ) $ is called an unbalanced connected NIS of $\mathcal{G}(\mathcal{V}, \mathcal{E})$.
\end{definition}

Next, let $\mathcal{E}_{\!c}$ be the set of edges connecting $\mathcal{G}_{\!s}$ and $\mathcal{G}_{\!m}(\! \mathcal{V}_{\!m}, \mathcal{E}_{\!m} \!) $ that denotes the NIS of $\mathcal{G}$ by nodes $\mathcal{V}_{\!m} \!\!=\!\! \mathcal{V}/\mathcal{V}_{\!s}$. Denote the set of endpoints of $\mathcal{E}_{\!c}$ in $\mathcal{V}_{\!s}$ by $\mathcal{V}_{\!cs}$ (called them \textit{boundary nodes} of the unbalanced connected NIS $\mathcal{G}_s$), and that in $\mathcal{V}_m$ by $\mathcal{V}_{cm}$. Then regarding graph $\mathcal{G}_m$ as a node denoted as $v_m$, $\mathcal{E}_c$ is mapped onto $\mathcal{E}_{cm}$ that replaces all endpoints in $\mathcal{V}_{cm}$ by $v_m$, so that we obtain graph $\mathcal{G}_{\!sm} (\mathcal{V}_{\!s} \cup \{v_m\}, \mathcal{E}_s \cup \mathcal{E}_{\!cm} ) $. Introduce graph $\mathcal{G}_{ms} (\mathcal{V}_m \cup \mathcal{V}_{cs}, \mathcal{E}_m \cup \mathcal{E}_{cs} ) $, where $\mathcal{E}_{cs} \!\subseteq\! \mathcal{V}_{cs} \!\times\! \mathcal{V}_{cs}$ is arbitrarily assigned as long as the NIS of $\mathcal{G}_{ms}$ by nodes $\mathcal{V}_{cs}$ is connected. Analogously to $\mathcal{G}_{\bm{z}}$, we use $\mathcal{G}_{*, \bm{z}}$ with $* \!\in\! \{s,m,sm,ms\}$ to denote the edge-induced subgraph of $\mathcal{G}_{*, \bm{z}}$ by edges $\{ k \!\in\! \mathcal{E}_* | z_k \!\!=\!\! 1 \}$ with $\mathcal{E}_{*}$ being the edge set of graph $\mathcal{G}_{*}$. Fig. \ref{fig-0-1-4} illustrates the graph notations introduced above.

\begin{figure}[h]
	\centering
	\includegraphics[width=7.5cm]{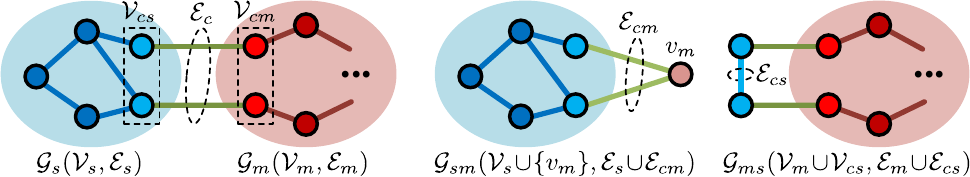}
	\caption{Illustration of some graph notations.}
	\label{fig-0-1-4}
\end{figure}

\begin{lemma}\label{lemma-0-1-con}
  Graph $\mathcal{G}_{\bm{z}}$ is connected iff $\mathcal{G}_{sm, \bm{z}}$ and $\mathcal{G}_{ms, \bm{z}}$ are both connected.
\end{lemma}
\begin{proof}
  The result follows directly from the definition of connected graphs. 
\end{proof}

\begin{theorem}\label{therom-0-1-2}
  Suppose $\bm{z}^{\star}$ is any feasible $\bm{z}$ for OTS model (\ref{eq-0-1-0}) and $\mathcal{G}_s(\mathcal{V}_s, \mathcal{V}_s)$ is an unbalanced connected NIS of $\mathcal{G}(\mathcal{V}, \mathcal{E})$. Then $\mathcal{G}_{\bm{z}^{\star}}$ is connected iff $\mathcal{G}_{ms, \bm{z}^{\star}}$ is connected.
\end{theorem}
\begin{proof} 
  Since $\bm{z}^{\star}$ is feasible for (\ref{eq-0-1-0}), by (\ref{eq-0-1-0:6})-(\ref{eq-0-1-0:4}) corresponding to $\mathcal{G}_s$, $\exists \bm{\varrho} \!\in\!  \{  \bm{p}_{gs} \!-\! \bm{p}_{ds}  | \bm{p}_{gs}^{\min} \!\leq\! \bm{p}_{\!gs} \!\leq\! \bm{p}_{\!gs}^{\max} \} $, solutions of the following equation exist:  
  \begin{equation}\label{eq-0-1-proof-1} 
    \bm{L}_{\mathcal{G}_{s, \bm{z}^{\star} }} \bm{\theta}_s = \bm{\varrho} + \bm{E}_{cs} \bm{Z}_{c}^{\star} \bm{B}_c (\bm{E}_{cs}^T \bm{\theta}_{s} - \bm{E}_{cm}^T \bm{\theta}_{cm})
  \end{equation} %+ \bm{E}_{sm} \bm{p}_{sm}
  where $\bm{L}_{\mathcal{G}_{s, \bm{z}^{\star} }}$ is the Laplacian matrix of $\mathcal{G}_{s, \bm{z}^{\star} }$ with $b_k$ being edge weights, $\bm{\theta}_{s} \!\!=\!\! [\theta_i]_{i \in \mathcal{V}_s}$, $\bm{\theta}_{cm} \!\!=\!\! [\theta_i]_{i \in \mathcal{V}_{cm}}$, $\bm{E}_{cs} \!\in\! \mathbb{R}^{|\mathcal{V}_{s}| \!\times\! |\mathcal{E}_{c}|}$ is the incidence matrix between $\mathcal{V}_s$ and $\mathcal{E}_c$, $\bm{E}_{cm} \!\in\! \mathbb{R}^{|\mathcal{V}_{cm}| \!\times\! |\mathcal{E}_{c}|}$ is the incidence matrix between $\mathcal{V}_{cm}$ and $\mathcal{E}_c$, $\bm{B}_c \!\!=\!\! \diag( b_k |_{k \in \mathcal{E}_c} )$, and $\bm{Z}_c^{\star} \!\!=\!\! \diag( z_k^{\star} |_{k \in \mathcal{E}_c} )$. 
  Denote by $(\bm{\theta}_s^{\star}, \bm{\theta}_{cm}^{\star})$ a solution of (\ref{eq-0-1-proof-1}). 
  Since $\mathcal{G}_s$ is an unbalanced connected NIS of $\mathcal{G}$, $\bm{1}_{|\mathcal{V}_s|}^T \bm{\varrho} \!\neq\! 0$. 
  Further with $ \bm{1}_{|\mathcal{V}_s|}^T\! \bm{L}_{\mathcal{G}_{s, \bm{z}^{\star} }}  \!\!=\!\! \bm{0}$, multiplying both sides of (\ref{eq-0-1-proof-1}) at solution $(\bm{\theta}_s^{\star}, \bm{\theta}_{cm}^{\star})$ by $\bm{1}_{|\mathcal{V}_s|}^T$ yields 
  \begin{equation}\label{eq-0-1-proof-3}
    \bm{1}_{|\mathcal{V}_s|}^T  \bm{E}_{cs} \bm{Z}_{c}^{\star} \bm{B}_c (\bm{E}_{cs}^T \bm{\theta}_{s}^{\star} - \bm{E}_{cm}^T \bm{\theta}_{cm}^{\star}) \!\!=\!\! -\! \bm{1}_{|\mathcal{V}_s|}^T \bm{\varrho} \!\neq\! 0 
  \end{equation} 
  Next, all entries of $\bm{\theta}_{cm}^{\star}$ are shifted to a same scaler $\theta_m^{\star} $ and $\bm{B}_c$ is transformed into $\bm{B}_c' \!\!=\!\! \diag( b_k' |_{k \in \mathcal{E}_c} )$ such that 
  \vspace{-2pt}
  \begin{equation}\label{eq-0-1-proof-2}
    \vspace{-2pt}
    \bm{B}_c (\bm{E}_{cs}^T \bm{\theta}_{s}^{\star} \!-\! \bm{E}_{cm}^T \bm{\theta}_{cm}^{\star}) = \bm{B}_c' (\bm{E}_{cs}^T \bm{\theta}_{s}^{\star} \!-\! \bm{E}_{cm}^T \bm{1}_{|\mathcal{V}_{cm}|} \theta_m^{\star} )  
  \end{equation}
  Then (\ref{eq-0-1-proof-1}) with equalities (\ref{eq-0-1-proof-3}) and (\ref{eq-0-1-proof-2}) gives
  \vspace{-2pt}
  \begin{equation}\label{eq-0-1-proof-4}
    \vspace{-2pt}
    \bm{L}_{\mathcal{G}_{sm, \bm{z}^{\star} }} [ {\bm{\theta}_s^{\star}}^T ~ \theta_m^{\star} ]^T
     = [ \bm{\varrho}^T ~  -\! \bm{1}_{|\mathcal{V}_s|}^T \bm{\varrho}]^T
  \end{equation}
  where $\bm{L}_{\mathcal{G}_{sm, \bm{z}^{\star} }}$ is the Laplacian matrix of $\mathcal{G}_{sm, \bm{z}^{\star} }$ with $b_k$ being weights for edges $\mathcal{E}_s$ and $b_k'$ being weights for edges $\mathcal{E}_c$. 
  Vector $[ \bm{\varrho}^T   -\! \bm{1}_{|\mathcal{V}_s|}^T \bm{\varrho}]^T$ is uniquely-balanced since $\bm{\varrho}$ is unbalanced. By Lemma \ref{lem-0-1-1} and (\ref{eq-0-1-proof-4}), graph $\mathcal{G}_{sm, \bm{z}^{\star}}$ is connected, which together with Lemma \ref{lemma-0-1-con} gives Theorem \ref{therom-0-1-2}.
\end{proof}

\vspace{-4pt}

By Theorem \ref{therom-0-1-2}, the first step for reduction of the connectedness constraints is 
\begin{itemize}
  \item[\textit{S.1:}] Find all unbalanced connected NISs of $\mathcal{G}(\!\mathcal{V},\! \mathcal{E}\!)$ and corresponding boundary node sets, denoted as $\{ \tilde{\mathcal{G}}_{\!s} (\! \tilde{\mathcal{V}}_{\!s}, \tilde{\mathcal{E}}_{\!s} \!) |_{s=1}^{N_u} \}$ and $\{ \mathcal{V}_{b, s} |_{s=1}^{N_u} \}$, respectively. Construct graph $\mathcal{G}_{o}(\mathcal{V}_o, \mathcal{E}_o )$, with $\mathcal{V}_o \!\!=\!\! (\mathcal{V} \!/ \cup_{s=1}^{N_{\!u}} \tilde{\mathcal{V}}_s) \cup (\cup_{s=1}^{N_{\!u}} \mathcal{V}_{b,s})  $, and $\mathcal{E}_o \!\!=\!\! (\mathcal{E} \!/ \cup_{s=1}^{N_u} \tilde{\mathcal{E}}_s) \cup (\cup_{s=1}^{N_u} \mathcal{E}_{b,s})  $, where $\mathcal{E}_{b,s} \!\subseteq\! \mathcal{V}_{b,s} \!\times\! \mathcal{V}_{b,s} $ is arbitrarily assigned as long as $\{ \mathcal{G}_{b,s}(\mathcal{V}_{b,s}, \mathcal{V}_{b,s} ) |_{s=1}^{N_u} \}$ are all connected. 
\end{itemize}

The next stage of the reduction is based on the fact that branches collected by $\mathcal{E}_u$ are unswitchable in OTS, i.e., partial entries of $\bm{z}$ are fixed to 1, which consist of three steps:
\begin{itemize}
  \item[\textit{S.2:}] Find all connected components of graph $\mathcal{G}_{\!u}(\! \mathcal{V}_o,\! (\!\mathcal{E}_u \!\cap\! \mathcal{E}_o) \!\cup\! (\cup_{\!s=1}^{\!N_u} \mathcal{E}_{\!b,s}) )$, denoted by $\mathcal{N} \!\!\!=\!\! \{\!\mathcal{G}_{\!u, k}\!(\!\mathcal{V}_{\!k},\! \mathcal{E}_k ) \!|_{\!k=1}^{N_c} \!\}$, and the set of edges connecting different components, denoted by $\mathcal{E}_l$.
  \item[\textit{S.3:}] Construct graphs $\{\mathcal{G}'_{u, k}(\mathcal{V}'_k, \mathcal{E}'_k ) |_{k=1}^{N_c} \} $ where $\mathcal{V}'_k \!\subseteq\! \mathcal{V}_k $ containing nodes being involved in $\mathcal{E}_{l}$, and $\mathcal{E}'_k$ is arbitrarily assigned as long as $\mathcal{G}'_{u, k}$ is connected.
  \item[\textit{S.4:}] Construct graph $\mathcal{G}'(\mathcal{V}', \mathcal{E}')$ with $\mathcal{V}' \!=\! \bigcup_{k=1}^{N_c} \mathcal{V}'_k $ and $\mathcal{E}' = \mathcal{E}_l \cup \bigcup_{k=1}^{N_c} \mathcal{E}'_k $.
\end{itemize}
\begin{figure}[h]
    \centering
	\includegraphics[width=8cm]{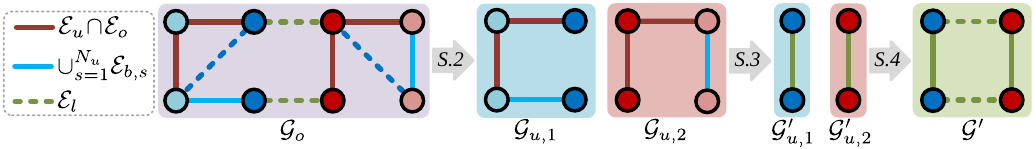}
	\caption{Illustration of steps for reduction of the connectedness constraints.}
	\label{fig-0-1-1}
\end{figure}
\textit{S.2}-\textit{S.4} are illustrated in Fig. \ref{fig-0-1-1}. Finally, since for any feasible $\bm{z}$ for (\ref{eq-0-1-0}), connectedness of the induced graphs from $\mathcal{G}$, $\mathcal{G}_o$ and $\mathcal{G}'$ is equivalent, the reduction ends as follows:
\begin{itemize}
  \item[\textit{S.5:}] Introduce vector $\bm{z}' \in \mathbb{B}^{|\mathcal{E}'|}$ to represent status of branches in $\mathcal{E}'$, whose entries equal to 1 for $\bigcup_{k=1}^{N_c} \! \mathcal{E}'_k$ and equal to corresponding entries in $\bm{z}$ for $\mathcal{E}_l$; % Let $\mathcal{G}'_{\bm{z}'}$ be the edge-induced subgraph of $\mathcal{G}'$ by edges $\{ e_i \in \mathcal{E}' | \bm{z}'_i = 1 \}$.
  \item[\textit{S.6:}] The reduced version of (\ref{eq-0-1-1}) is obtained by replacing $\bm{E}_{\mathcal{G}}, \bm{\vartheta}, \bm{\rho}, \bm{1}_{N_e}$ and $\bm{c} $ in (\ref{eq-0-1-1}) by their counterparts for graph $\mathcal{G}'$, and $\bm{z}$ by $\bm{z}'$.
\end{itemize}

\section{Case Study}

The proposed connectedness constraints are demonstrated by adding them to model (\ref{eq-0-1-0}) \hig{and the $N\!-\!1$ security-constrained DC OTS model} both with $f$ being the total generation cost for two systems: IEEE 30-bus system and German transmission network in SciGRID \cite{4-1125}. We use M1 to M4 referring to the original OTS model (\ref{eq-0-1-0}), (\ref{eq-0-1-0}) with (\ref{eq-0-1-nc}), (\ref{eq-0-1-1}) and reduced (\ref{eq-0-1-1}) added, respectively; \hig{and N1 to N4 referring to the $N\!-\!1$ security-constrained version of M1 to M4, respectively.} For each system, $\lceil\! \alpha N_{\!e} \!\rceil$ lines are assumed to be switchable with $\alpha \!\in\! \{0.3 \!+\! 0.1 k|_{k=0}^4  \}$, where $\lceil \cdot \rceil$ represents ceiling, \hig{and $\mathcal{C}$ contains every branch in $\mathcal{E}$ whose outage causes no islands of $\mathcal{G}(\mathcal{E}, \mathcal{V})$.} For the IEEE 30-bus system, 100 different configurations of switchable lines are generated randomly; and for the German transmission network, 100 sample points from the time-series data are used while the configuration of switchable lines is fixed. Gurobi 9.0 is used to solve MILPs with default solver parameters on a Linux 64-Bit PC with an Intel(R) Core(TM) i5-6500 CPU $\!$@$\!$ 3.20GHz and 16GB RAM \hig{for M1 to M4, and on a Linux 64-Bit server with 2 Intel(R) Xeon(R) CPUs E5-2640 v4 $\!\!$@$\!\!$ 2.40GHz and 125GB RAM for N1 to N4.}

% The total number of lines to be switched off is bounded by $\lceil\! 0.05 N_e \!\rceil$ and $\lceil\! 0.15 N_e \!\rceil$. All data and source code can be found in \cite{4-971}. 

\begin{table}[H]
	\centering
  \caption{\hig{Number of Connected Optimal Topologies.}}
  \vspace{-10pt}
  \hig{
  {\footnotesize{
    \begin{tabular*}{\hsize}{@{}p{0.5cm}  @{}p{0.523cm}@{}p{0.523cm}@{}p{0.523cm}@{}p{0.523cm}@{}p{0.523cm}@{}p{0.523cm}@{}p{0.523cm}@{}p{0.523cm}@{}p{0.523cm}@{}p{0.523cm}@{}p{0.523cm}@{}p{0.523cm}@{}p{0.523cm}@{}p{0.523cm}@{}p{0.523cm}@{}p{0.3cm}    }
      \toprule
      \multirow{2}{*}{$\alpha$} & \multicolumn{8}{c}{IEEE 30-bus system} & \multicolumn{8}{c}{German transmission network} \\ \arrayrulecolor{darkgray}\cline{2-9}  \arrayrulecolor{lightgray}\cline{10-17}
       \multicolumn{1}{c}{}  & M1 & N1 &  M2 & N2 &  M3 & N3 &  M4 & N4 &  M1 & N1 &  M2 & N2 &  M3 & N3 &  M4 & N4 \\  [-1.5pt] \arrayrulecolor{black}\midrule 
       0.3 & 65 & 24 & 98  & 53  & 100 & 100  & 100 & 100  & 10 & 14  & 67 & 28 & 100 & 100 & 100 & 100    \\     
       0.4 & 73 & 35 & 97  & 46  & 100 & 100  & 100 & 100  & 7  & 9   & 58 & 20 & 100 & 100 & 100 & 100    \\ 
       0.5 & 72 & 27 & 97  & 55  & 100 & 100  & 100 & 100  & 6  & 9   & 60 & 19 & 100 & 100 & 100 & 100   \\
       0.6 & 73 & 18 & 99  & 43  & 100 & 100  & 100 & 100  & 7  & 5   & 54 & 25 & 100 & 100 & 100 & 100   \\
       0.7 & 75 & 13 & 100 & 36  & 100 & 100  & 100 & 100  & 5  & 0   & 64 & 17 & 100 & 100 & 100 & 100    \\
       \arrayrulecolor{black}\bottomrule 
      \end{tabular*}
      \scriptsize{\vspace{0mm} \\ Note: For each value of $\alpha$ and OTS model, the maximal number of connected optimal topologies equals to the number of configurations of switchable lines or sample points (i.e., 100). For N1 to N4, the optimal topology is considered connected iff itself and all\\ \vspace{-1.5mm}its post-contingency topologies are connected. \textcolor{white}{text  1 1 1 1 1 1 111 11 1 111  1re.. } 
     }   
  }}
  }
  \vspace{-5pt}
	\label{table-0-1-1}
\end{table}
 
\begin{figure}[H]
	\centering
	\includegraphics[width=8.8cm]{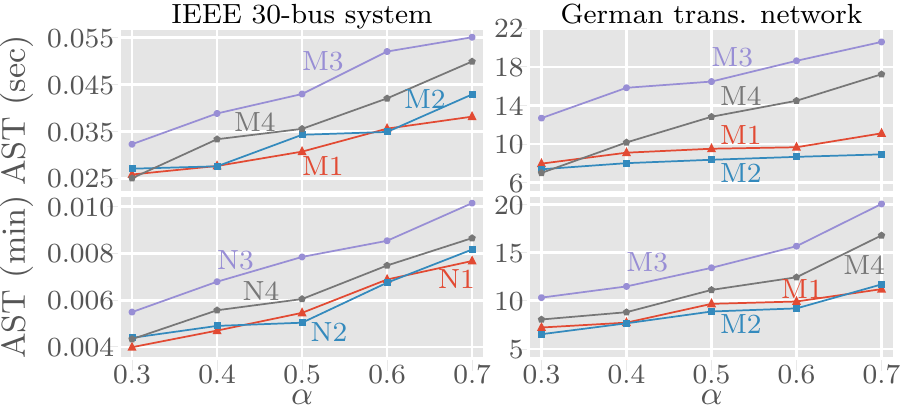}
	\caption{\hig{Average solution time (AST) for M1 to M4 and N1 to N4 with varying $\alpha$. The AST is calculated by averaging the solution time of corresponding 100 line configurations or sample points.} }
	\label{fig-0-1-r3}
\end{figure}
 
Table \ref{table-0-1-1} compares the number of connected optimal topologies obtained by different OTS models with varying $\alpha$. \hig{It is found that for the IEEE 30-bus system without $N\!-\!1$ security,} connectedness of the optimal topology may be ensured without connectedness constraints, and can almost always be ensured with only the necessary constraints. However, for the German transmission network, under all values of $\alpha$, no more than 10\% optimal topologies are connected for M1, and around 40\% obtained optimal topologies are unconnected for M2. \hig{For both systems, ensuring connectedness generally becomes much harder for N1 and N2.} In contrast, with our proposed connectedness constraints or the reduced version, connectedness of the optimal topology is always ensured for both systems. \hig{It shows the necessity of including connectedness constraints into OTS models, especially for large systems with $N\!-\!1$ security constraints and varying operating conditions.}

\hig{Fig. \ref{fig-0-1-r3} gives the average solution time for each OTS model with varying $\alpha$.} It is observed that introducing the proposed connectedness constraints increases solution time compared with that of \hig{M1 and M2 or N1 and N2}, which however, is moderate and can be remarkably relieved by using the reduced version. In particular, for both test system with $\alpha = 0.3$, solution time when using the reduced connectedness constraints is \hig{close to that of M1 and M2 or N1 and N2}. It is worth pointing out that all these observations are obtained under the DC OTS model, while for more complex AC OTS where mixed-integer nonlinear programs are solved, computational burdens caused by the proposed linear constraints should be more moderate.

\section{Conclusion}

In this paper, we filled a long-existing gap of OTS problems: how to strictly ensure network connectedness in an efficient way when solving optimization models of OTS. By adding the proposed set of linear constraints to optimization models of OTS, connectedness of the optimal topology can be strictly ensured while only moderate computational burdens are caused, which can be further relieved by using the proposed reduced version. The proposed approach has the potential to be extended to some other topology optimization and control problems in power grids to impose network topology. The key of the extension should be design $\bm{c}$ to satisfy certain conditions instead of being uniquely-balanced, such that all eligible topologies exactly correspond to the feasible domain of $\bm{z}$ of the proposed linear constraints. 

Future work focuses on further improving computational efficiency of OTS models with strict connectedness constraints. Firstly, the numerical results show that adding constraints (\ref{eq-0-1-nc}) is unable to ensure connectedness strictly while for large-scale systems, can slightly reduce solution time. Thus combining (\ref{eq-0-1-nc}) and a reduced version of (\ref{eq-0-1-1}) can promisingly bring computationally cheaper yet strict connectedness guarantee. Secondly, values of $M$ and $\bm{c}$ also can impact solutions time, which deserves a theoretical analysis to guide their value selection.

\ifCLASSOPTIONcaptionsoff
  \newpage
\fi

\bibliographystyle{IEEEtran}
\bibliography{References/4}

\end{document}